\pdfoutput=1
\RequirePackage{ifpdf}
\ifpdf 
\documentclass[pdftex]{sigma}
\else
\documentclass{sigma}
\fi

\numberwithin{equation}{section}

\def\be{\begin{array}{l}}
\def\ee{\end{array}}
\def\beqa{\begin{eqnarray}}
\def\eeqa{\end{eqnarray}}

\newtheorem{Theorem}{Theorem}[section]
\newtheorem{Corollary}[Theorem]{Corollary}
{ \theoremstyle{definition}
\newtheorem{Remark}[Theorem]{Remark} }

\begin{document}

\allowdisplaybreaks

\newcommand{\arXivNumber}{1803.06819}

\renewcommand{\thefootnote}{}

\renewcommand{\PaperNumber}{106}

\FirstPageHeading

\ShortArticleName{Generalized Hermite Polynomials and Monodromy-Free Schr\"odinger Operators}

\ArticleName{Generalized Hermite Polynomials\\ and Monodromy-Free Schr\"odinger Operators\footnote{This paper is a~contribution to the Special Issue on Painlev\'e Equations and Applications in Memory of Andrei Kapaev. The full collection is available at \href{https://www.emis.de/journals/SIGMA/Kapaev.html}{https://www.emis.de/journals/SIGMA/Kapaev.html}}}

\Author{Victor Yu.~NOVOKSHENOV}

\AuthorNameForHeading{V.Yu.~Novokshenov}

\Address{Institute of Mathematics, Russian Academy of Sciences,\\ 112 Chernyshevsky Str., 450008, Ufa, Russia}
\Email{\href{mailto:novik53@mail.ru}{novik53@mail.ru}}
\URLaddress{\url{http://matem.anrb.ru/ru/novokshenovvy}}

\ArticleDates{Received March 20, 2018, in final form September 20, 2018; Published online September 30, 2018}

\Abstract{The paper gives a review of recent progress in the classification of mo\-no\-dro\-my-free Schr\"odinger ope\-ra\-tors with rational potentials. We concentrate on a class of potentials constituted by generalized Hermite polynomials. These polynomials defined as Wronskians of classic Hermite polynomials appear in a number of mathematical physics problems as well as in the theory of random matrices and 1D SUSY quantum mechanics. Being quadratic at infinity, those potentials demonstrate localized oscillatory behavior near the origin. We derive an explicit condition of non-singularity of the corresponding potentials and estimate a localization range with respect to indices of polynomials and distribution of their zeros in the complex plane. It turns out that 1D SUSY quantum non-singular potentials come as a~dressing of the harmonic oscillator by polynomial Heisenberg algebra ladder operators. To this end, all ge\-ne\-ralized Hermite polynomials are produced by appropriate periodic closure of this algebra which leads to rational solutions of the Painlev\'e~IV equation. We discuss the structure of the discrete spectrum of Schr\"odinger operators and its link to the monodromy-free condition.}


\Classification{30D35; 30E10; 33C75; 34M35; 34M55; 34M60}

\rightline{\it In memory of my friend and colleague Andrei Kapaev}

\renewcommand{\thefootnote}{\arabic{footnote}}
\setcounter{footnote}{0}

\section{Introduction}\label{sect0}

The integrability property of Painlev\'e equations reveals a number of applications of their solutions. Besides traditional self-similar modes in nonlinear PDE's of mathematical physics they provide new construction material for integrable quantum mechanics and spectral theory. In this paper, we give a brief review of recent achievements in these applications of rational solutions of the fourth Painlev\'e equation (PIV). We trace how they come from monodromy-free potentials of the Schr\"odinger equation and from supersymmetric dressing of the harmonic potential in one-dimensional quantum mechanics.

Another ingredients of these interconnections are the generalized Hermite polynomials which build all rational solutions of PIV. Their appearance in multiple applications is explained by the determinant representation of these polynomials. Actually, it can be set as a definition in terms of classical Hermite polynomials. Namely, generalized Hermite polynomials (GHP) $H_{m,n}(z)$ are defined as follows~\cite{clarkson,NY}
\begin{gather*}
H_{m,n}(z)=\det\left( P_{n-i+j}(z)\right)_{i,j=1}^{m},
\end{gather*}
 where
\begin{gather*}
P_s(z)= \sum_{i+2j=s} \frac{1}{6^j i! j!} z^i,
\end{gather*}
or, equivalently, as Wronskians of classical Hermite polynomials
\begin{gather}\label{ghp}
H_{m,n}(z)= c_{m,n}{\mathcal W}\left(H_m(z), H_{m+1}(z), \ldots, H_{{{m+n-1}}}(z)\right),
\end{gather}
where $H_n(z)=(-1)^n {\rm e}^{z^2}\frac{{\rm d}^n}{{\rm d}z^n} {\rm e}^{-z^2}$ and $c_{m,n}$ are normalization constants.

Like classical orthogonal polynomials $H_{m,n}$ have a number of useful properties. For example, they constitute recurrence coefficients for orthogonal polynomials $p_n(x)$ with weight $w(x,z,m)=(x-z)^m \exp\big({-}x^2\big)$ \cite{cf,deift}
\begin{gather*}
xp_n(x)=p_{n+1}(x)+a_n(z,m)p_n(x)+b_n(z,m)p_{n-1}(x),
\end{gather*}
where
\begin{gather*}
a_n(z,m)=-\frac 12\frac{{\rm d}}{{\rm d}z}\ln\frac{H_{n+1,m}}{H_{n,m}}, \qquad b_n(z,m)=\frac{nH_{n+1,m}H_{n-1,m}}{2H^2_{n,m}}.
\end{gather*}
Another property is a formula for rational solutions to the Painlev\'e IV equation
\begin{gather}\label{p4sol}
v(z)={-2z+}\frac{{\rm d}}{{\rm d}z}\ln\frac{H_{{m,n+1}}(z)}{H_{{m+1,n}}(z)}.
\end{gather}
In this case PIV equation
\begin{gather}
v''=\frac{(v')^2}{2v}+\frac 32 v^3 + 4zv^2+2\big(z^2-a\big)v+{\frac{b}{v}}, \label{p4}
\end{gather}
has integer coefficients
\begin{gather*}
a = n-m, \qquad b = -2(m+n+1)^2,
\end{gather*}
where $m$ and $n$ are the indices of the corresponding polynomials~\cite{luk}.

 The solutions \eqref{p4sol} have a specific structure of poles in the complex plane. It can be thought of as an equilibrium state of Coulomb charged particles in an external field. Indeed, any pole of a rational solution to PIV has residue equal to $c_j=+1$ or $c_j=-1$ (see Theorem~\ref{th1}). The poles~$z_j$ can be interpreted as positive and negative charges interacting by the logarithmic potential and influenced by the external quadratic potential
 \begin{gather*}
 U(z_1, z_2, \ldots, z_n) = \sum\limits_{j=1}^nc_jz_j^2+ \sum\limits_{j\not= k}^nc_jc_k\log(z_j-z_k)^2.
 \end{gather*}
 The equilibrium condition provides the {\em generalized Stiltjes relation} \cite{V}
 \begin{gather*}
 \sum\limits_{j\not= k}\frac{c_j}{z_k-z_j}+z_k=0, \qquad k=1,2,\ldots.
 \end{gather*}
 Here each pole coincides with zero of the related polynomial. The distribution of poles for large orders of polynomials has been studied since classical works by T.~Stiltjes~\cite{st} and M.~Plancherel and W.~Rotach~\cite{pr}. This question was discussed recently in applications to dynamics of Coulomb log-gases~\cite{for} and approximations by rational functions in logarithmic potential theory~\cite{saff}.

Since the PIV equation \eqref{p4} is integrable in the sense of soliton theory \cite{fikn,kapaev1}, all its rational solutions have been found and labeled by recursions of B\"acklund transformations~\cite[Chapter~6]{fikn}. These recursions can be rewritten as dressing chains of the Lax operator with some periodic closures. As a by-product this gives a set of Schr\"odinger operators $L$ formed by GHP
\begin{gather}
 L = -\frac{{\rm d}^2}{{\rm d}x^2} + u(x), \qquad
 u(x) = f'(x) +f^2(x),\qquad f(x) = -x+\frac{{\rm d}}{{\rm d}z}\ln\frac{H_{m,n}(z)}{H_{m,n+1}(z)}.\label{mf}
\end{gather}
These operators are monodromy-free (see Section~\ref{sect1}) and the discrete spectrum of each is an arithmetic sequence with a finite gap (Theorem~\ref{th3}). Moreover, all potentials~\eqref{mf} are non-singular on the real line for odd~$n$. This is due to Theorem~\ref{th2} below which proves $\operatorname{Res} u(z) =0$. In turn, this follows from the distribution of zeros of GHP.

One can mention also a recent application of GHP in matrix models of statistical physics. Consider a degenerate {\em Gaussian unitary ensemble} where eigenvalues $\lambda_k$ are fixed $k=1,2, \ldots, n$, and $\lambda_{n+1}=z$ has $m$-fold multiplicity. Then the partition function of the ensemble has the form~\cite{cf}
\begin{gather}\label{gue}
 D_n(z)=\frac{1}{n!}\int^\infty_\infty \cdots \int^\infty_\infty \prod_{1\le i<j\le n} (\lambda_i - \lambda_j)^2\prod\limits_{k=1}^n (\lambda_k -z)^m {\rm e}^{-\lambda_k^2}{\rm d}\lambda_k,
\end{gather}
where
\begin{gather*}
 D_n(z) = A_{m,n}H_{m,n}(c z), \qquad c = {\rm i}\sqrt{\frac 23}, \qquad A_{m,n}= {\rm const}.
\end{gather*}
Note that formula \eqref{gue} is proved with the help of dressing chains and ladder operators discussed below in Section~\ref{sect4}.

\begin{figure}[t]\centering
\includegraphics[width=70mm]{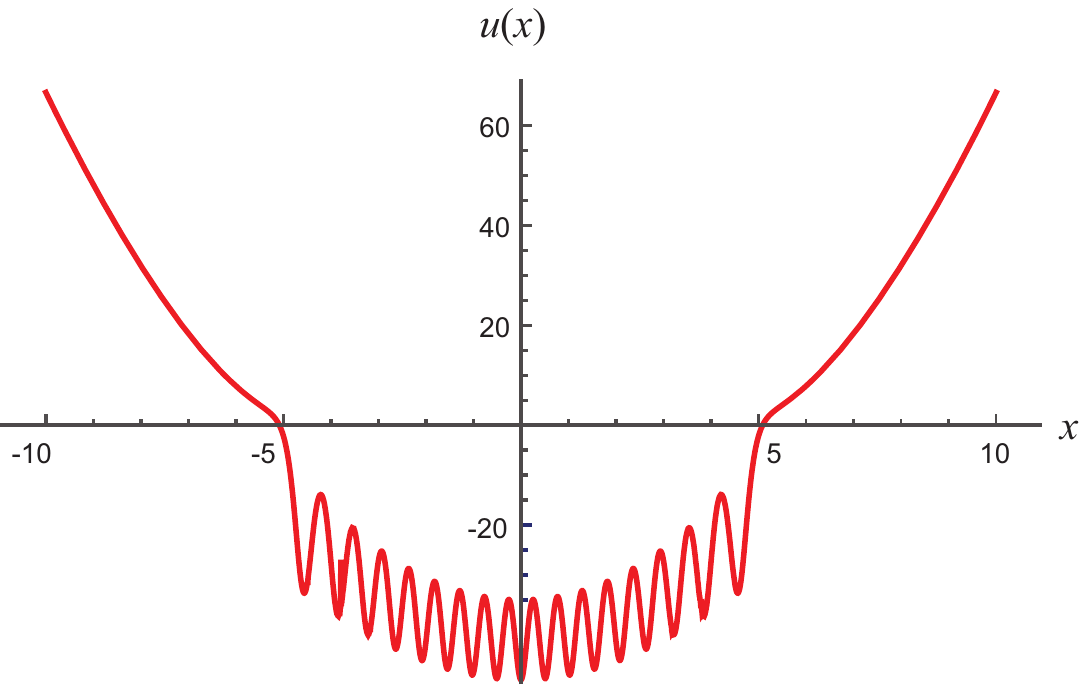} \qquad \includegraphics[width=70mm]{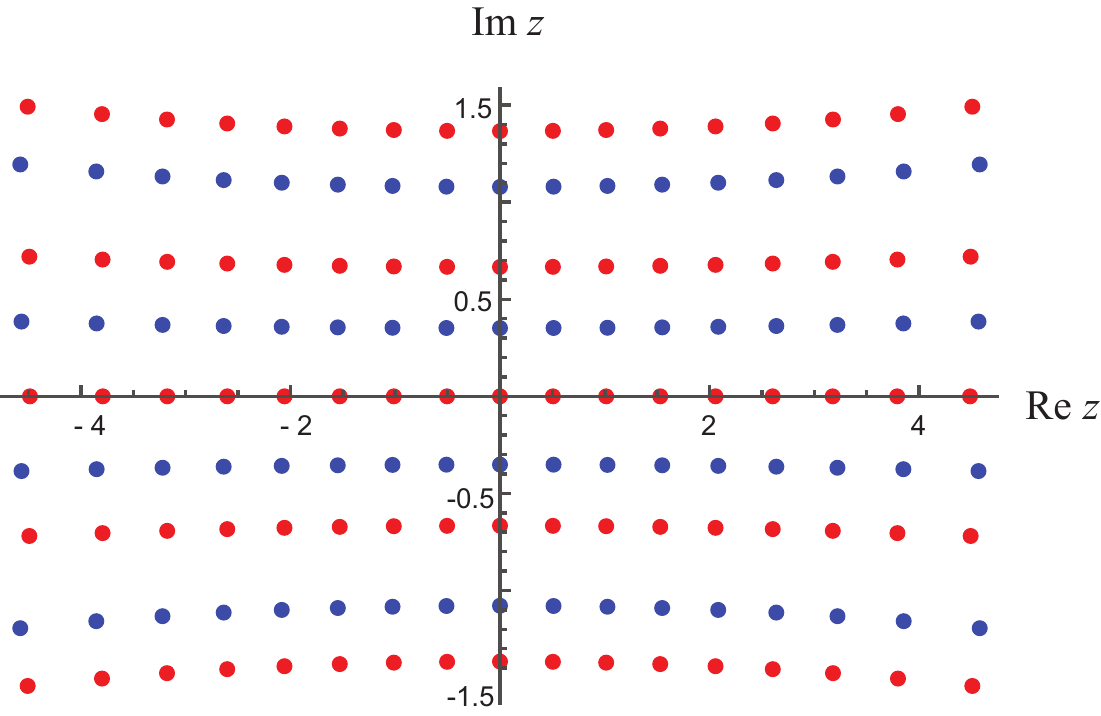}
\caption{Monodromy-free potential~\eqref{mf} $u(x)$, generated by polynomials $H_{m,n}(z)$ and $H_{m,n+1}(z)$ for $m=17$, $n=4$ (left) and zeros of $H_{m,n}$ (blue) and $H_{m,n+1}$ (red) in the complex plane (right).} \label{fig1}
\end{figure}

\section{Monodromy-free potentials and dressing chain}\label{sect1}

A Schr\"odinger operator
\begin{gather}\label{sch}
L = -\frac{{\rm d}^2}{{\rm d}z^2} + u(z)
\end{gather}
with meromorphic potential is called {\em monodromy-free} if all solutions of the equation
\begin{gather*}
L\psi = -\psi'' + u\psi= \lambda\psi
\end{gather*}
are meromorphic in the whole complex plane $z\in\mathbb C$ for all $\lambda$. In other words, monodromy of the equation \eqref{sch} in the complex plane is trivial for all $\lambda$.

The problem of classification of monodromy-free Schr\"odinger operators is traditional in spectral theory (see~\cite{DGr,MKTrub}). It became even more important in soliton theory where monodromy-free potentials form a class of soliton solutions to nonlinear equations for which a~Schr\"o\-din\-ger operator enters the Lax pair. In a case of potentials decreasing at infinity on the real line the spectral theory is well understood (see~\cite{thsol}). In the framework of soliton theory there were found new classes of monodromy-free potentials such as finite-gap ones and rational potentials with quadratic growth at infinity. Here the latter class will be studied in detail in the special case of rational solutions of the PIV equation.

According to \cite{VS}, the Schr\"odinger operator \eqref{sch} $L_j$ with potential $u_j(z)$ is factorized in the form
\begin{gather}\label{sch1}
 L_j = -\left(\frac{{\rm d}}{{\rm d}z} + f_j(z)\right)\left(\frac{{\rm d}}{{\rm d}z} - f_j(z)\right),
\end{gather}
where the function $f_j(z)$ satisfies the Riccati equation ($' ={\rm d}/{\rm d}z$)
 \begin{gather}\label{ric}
f_j' + f_j^2 = u_j.
 \end{gather}
The Darboux transformation
 \begin{gather*}
 L_j \ \mapsto \ L_{j+1}= -\left(\frac{{\rm d}}{{\rm d}z} - f_j(z)\right)\left(\frac{{\rm d}}{{\rm d}z} + f_j(z)\right)+\alpha_j,
\end{gather*}
produces the new potential
\begin{gather*}
u_{j+1} = f'_{j+1} + f^2_{j+1} = -f_j'+ f_j^2 +\alpha_j = u_j -2f_j' +\alpha_j.
 \end{gather*}
This gives rise to the {\em dressing chain} equations \cite{VS}
\begin{gather}\label{drchain}
f'_{j+1} + f'_j = f_j^2 - f_{j+1}^2 + \alpha_j, \qquad j = 1, 2, \ldots, n, \ldots,
\end{gather}
where $\alpha_i$ are arbitrary constants.

In other words, equations \eqref{drchain} are equivalent to the relations between Schr\"o\-din\-ger operators
\begin{gather*}
-\left(\frac{{\rm d}}{{\rm d}z} - f_j(z)\right)\left(\frac{{\rm d}}{{\rm d}z} +f_j(z)\right) +\alpha_j =-\left(\frac{{\rm d}}{{\rm d}z} + f_{j+1}(z)\right)\left(\frac{{\rm d}}{{\rm d}z} - f_{j+1}(z)\right).
\end{gather*}
This property plays a key role in the calculation of spectrum of the monodromy-free potential~\eqref{ric} $u(x) = f_1'(x) +f^2_1(x) $.

Following \cite{adler,bur80,VS}, consider the following periodic closure of the dressing chain~\eqref{drchain}
\begin{gather*}
f_j=f_{j+N}, \qquad \alpha_j = \alpha_{j+N}.
\end{gather*}
For $N=3$ the infinite chain \eqref{drchain} reduces to the second-order ODE written in symmetric form (sPIV) found in \cite{adler,bur80}
\begin{gather}
\phi'_1 + \phi_1(\phi_2 - \phi_3) - \alpha_1=0, \qquad
\phi'_2 + \phi_2(\phi_3 - \phi_1) - \alpha_2=0,\nonumber\\
\phi'_3 + \phi_3(\phi_1 - \phi_2) - \alpha_3=0,\label{adl}
\end{gather}
where
\begin{gather*}
\phi_1 = f_1 + f_2, \qquad \phi_2 = f_2 + f_3, \qquad \phi_3 = f_3 + f_1,
\end{gather*}
and
\begin{gather*}
\phi_1 + \phi_2 + \phi_3 = -2z, \qquad \alpha_1+ \alpha_2 + \alpha_3 = -2.
\end{gather*}
The system \eqref{adl}, in turn, is equivalent to the PIV equation \cite{adler}
\begin{gather*}
 v''=\frac{(v')^2}{2v}+\frac 32 v^3 + 4zv^2+2\big(z^2-a\big)v+{\frac{b}{v}},
\end{gather*}
with
\begin{gather*}
v = \phi_1, \qquad a = \frac 12(\alpha_{{2}} - \alpha_{{3}}), \qquad b = - \frac 12 \alpha_1^{{2}}.
\end{gather*}

The class of rational solutions to the system sPIV has simple ``seed solutions''
\begin{gather}\label{hier}
\phi_1 = -\frac 1z, \qquad \phi_2 = \frac 1z, \qquad \phi_3 = -2z,
\end{gather}
with parameters $\alpha_1=\alpha_2={-2}$ and $\alpha_3= {0}$. They correspond to the ``$-1/z$ hierarchy'' and the ``$-2z$ hierarchy'' of PIV rational solutions first found by N.~Lukashevich~\cite{luk}. He proved that there are no other PIV rational solutions with leading terms $-1/z$ or $-2z$ at infinity. Note that there is also a~``$-\frac 23 z$'' hierarchy of PIV generated by Okamoto polynomials~\cite{NY} which we will not consider here.

The ``$-1/z$'' and ``$-2z$'' hierarchies correspond to rational solutions of the sPIV equa\-tion~\eqref{adl} formed by generalized Hermite polynomials $H_{m,n}$ \cite{clarkson}
\begin{gather}
 \phi_1(z) = \frac{{\rm d}}{{\rm d}z} \ln \frac{H_{m+1,n}(z)}{H_{m,n}(z)}, \qquad \phi_2(z) = \frac{{\rm d}}{{\rm d}z} \ln \frac{H_{m,n}(z)}{H_{m,n+1}(z)},\nonumber\\
\phi_3(z) = -2z+\frac{{\rm d}}{{\rm d}z} \ln \frac{H_{m,n+1}(z)}{H_{m+1,n}(z)}, \label{3p4}
\end{gather}
where $\alpha_1 = -2n$, $\alpha_2 = {2m +2n}$, $\alpha_3 = {-2m -2}$.

Due to the symmetry of the sPIV system \eqref{adl} and the periodic relations $f_j=f_{j+3}$ the first dressing chain component takes the form
\begin{gather}
 f^{(1)}_1(z) = -z - \frac{{\rm d}}{{\rm d}z} \ln \frac{H_{m,n}(z)}{H_{m,n+1}(z)}, \qquad f^{(2)}_1(z) = z + \frac{{\rm d}}{{\rm d}z} \ln \frac{H_{m+1,n}(z)}{H_{m,n+1}(z)},\nonumber\\
f^{(3)}_1(z) = -z {+} \frac{{\rm d}}{{\rm d}z}\ln \frac{H_{m,n}(z)}{H_{m+1,n}(z)}. \label{f1}
\end{gather}

Note that formulas \eqref{f1} can be derived also from results of A.~Oblomkov \cite{oblom}. He proved that any monodromy-free potential of a~Schr\"odinger operator \eqref{sch1} with
\begin{gather*}
f_j=\sum\limits_{k=1}^N\frac{c_k}{z-z_k} -z
\end{gather*}
is quadratic at infinity and has the form
\begin{gather}\label{obl}
u(z) = z^2-2\frac{{\rm d}^2}{{\rm d}z^2}\ln\mathcal{W} (H_m(z), H_{m+1}(z), \ldots, H_{m+n{-1}}(z) ),
\end{gather}
where $\mathcal{W}$ is the Wronskian and $H_k$ are classical Hermite polynomials. This form of the potential can easily be derived from the definition $u = f_1^\prime +f_1^2$ by using relations~\eqref{f1}
 \begin{subequations}\label{fla-flc}
 \begin{gather}
u(z) = z^2 - 2\frac{{\rm d}^2}{{\rm d}z^2} \ln H_{m+1,n}(z) +2n-1, \label{f1a} \\
u(z) = z^2 - 2\frac{{\rm d}^2}{{\rm d}z^2} \ln H_{m,n+1}(z) + 2n - 2m +1, \label{f1b}\\
u(z) = z^2 - 2\frac{{\rm d}^2}{{\rm d}z^2} \ln H_{m,n}(z) -2m-1. \label{f1c}
\end{gather}
 \end{subequations}
Taking into account the definition \eqref{ghp} we come to the representation \eqref{obl} up to a constant term.

\section{Non-singular potentials on the real line}\label{sect3}

Spectral theory of Schr\"odinger operators \eqref{sch1} usually supposes a potential $u(x)$ being non-singular, i.e., belonging to some functional space like~$L_2(\mathbb{R})$. This is true especially for applications like quantum mechanics as we will discuss in Section~\ref{sect4}.

We begin with the structure of poles of rational PIV solutions \eqref{f1}. In general, all solutions of~\eqref{p4} are meromorphic functions. They are described by the following

\begin{Theorem}[\cite{V}]\label{th1} Any rational solution to equation \eqref{p4} has the form
\begin{gather}\label{pol_p4}
v(z)=\epsilon z+ \sum\limits_j\frac{c_j}{z-z_j}, \qquad \epsilon=0, -\frac 23, -2, \qquad c_j=\pm 1, \qquad j=1,2, \ldots,
\end{gather}
and the {\em generalized Stieltjes relation} is true
\begin{gather}\label{stl}
\sum\limits_{j\not= k}\frac{c_j}{z_k-z_j}+(\epsilon +1)z_k=0, \qquad k=1,2,\ldots.
\end{gather}
\end{Theorem}

\begin{proof} Take a Laurent series near a pole $z=z_j$
\begin{gather}\label{pol_ser}
 v(z)=\sum\limits_{k=-l}^{\infty}c_k (z-z_j)^k, \qquad z\to z_j.
\end{gather}
Balancing the leading terms of the series in equation (\ref{p4}) yield $l=1$, $c_{-1}^2=1$. A~similar comparison at infinity proves that rational solution $u(z)$ has at most linear growth $u(z) =\epsilon z +O(1)$ as $z\to\infty$ with $\epsilon = 0, -2/3, -2$. Expand a rational solution~$u(z)$ into simple frac\-tions~(\ref{pol_p4}) and put it into equation (\ref{p4}) looking for terms of order $O(z-z_k)^{-2}$ as $z\to z_k$. Balancing these terms gives the relations (\ref{stl}) for any pole $z_k$.
\end{proof}

\begin{Corollary}\label{cor1} For any solution $v$ of the PIV equation \eqref{p4} the residues of the function $(z + v(z))^2$ are zero at any pole $z=z_k$
\begin{gather*}
\operatorname{Res} (z + v(z))^2 = 0.
\end{gather*}
\end{Corollary}

\begin{proof} From the Laurent series \eqref{pol_ser} one easily derives $c_0 = -z_j$. This yields the similar asymptotics for $z+v(z)$
\begin{gather*}
z+v(z) = \frac{c_{-1}}{z-z_j} + (c_0 + z_j) + (c_1+1)(z-z_j) + \cdots \\
\hphantom{z+v(z)}{} = \frac{c_{-1}}{z-z_j} + (c_1+1)(z-z_j) + \cdots.\tag*{\qed}
\end{gather*}\renewcommand{\qed}{}
\end{proof}

Each pole of the functions \eqref{f1} comes from a zero of a~GHP $H_{m,n}$, $H_{m,n+1}$ or $H_{m+1,n}$. The polynomials satisfy the recurrence
relations~\cite{clarkson}
\begin{gather}
2m H_{m+1,n}H_{m-1,n} =H_{m,n}H^{\prime\prime}_{m,n} - (H^\prime_{m,n})^2 + 2mH^2_{m,n},\nonumber\\
2n H_{m,n+1}H_{m,n-1} = - H_{m,n}H^{\prime\prime}_{m,n} + (H^\prime_{m,n})^2+ 2n H^2_{m,n},\label{recur}
\end{gather}
with initial conditions
\begin{gather*}
H_{0,0} = H_{0,1} = H_{1,0} = 1, \qquad H_{1,1}= 2z.
\end{gather*}
One can easily prove by induction that solutions of the system~\eqref{recur} are polynomials and each~$H_{m,n}$ has $mn$ simple zeros. Due to the obvious symmetries{\samepage
\begin{gather*}
H_{m,n}(-z) = (-1)^{mn}H_{m,n}(z), \qquad H_{m,n}({\rm i}z) = {\rm i}^{mn}H_{{n,m}}(z),
\end{gather*}
all zeros form a~symmetric pattern with respect to real and imaginary axes.}

Note that $H_{m,1} = H_m$, where $H_m$ is the classical Hermite polynomial of $m$-th order, $H_m(z)=(-1)^m {\rm e}^{z^2}\frac{{\rm d}^m}{{\rm d}z^m} {\rm e}^{-z^2}$. This means that all zeros of $H_{m,1}$ are on the real line. However, all poly\-no\-mials~$H_{m,n}$ with even $n$ do not have any real-valued zeroes. This follows from the theorem of V.E.~Adler
 \begin{Theorem}[\cite{adler2}] \label{th5}
For $x\in\mathbb{R}$ all Wronskians ${\mathcal W} (H_{m_1}(x), H_{m_2}(x), \ldots, H_{m_n}(x) ) \not= 0$ if $m_1 < m_2 < \dots < m_n$ and $n$ is even.
\end{Theorem}

The typical pattern of zeros is a slightly deformed rectangular region shown in Fig.~\ref{fig1} (right). Its horizontal and vertical ranges are proportional to $\sqrt{2m+n}$ and $\sqrt{2n+m}$ respectively. Recently, the generalized Hermite polynomials have been studied in the limit $m,n\to\infty$ in a~number of papers \cite{buck,FHV, mr,nov_sch}. In the paper~\cite{mr}, the distribution of zeros of $H_{m,n}(z)$ was obtained in the asymptotic limit $m\to\infty$, $n=O(1)$. On the other hand, the paper~\cite{buck} contains an analysis of $H_{m,n}(z)$ in the limit $m,n\to\infty$, $m=rn$, $r=O(1)$ and deduces in particular bounds for the deformed rectangular region
containing the zeros. The asymptotic distribution of zeros in the latter asymptotic regime, i.e., the generalization of Plancherel--Rotach formulas to~$H_{m,n}(z)$ for~$z$ within the deformed rectangular region for $m$ and $n$ of similar large order, remains an open problem. We note that both papers \cite{buck, mr} apply methods of asymptotic ``undressing'' of Riemann--Hilbert problems.

\begin{Remark} Note that in \cite{buck} a mistake was corrected in the asymptotics found in \cite{nov_sch}. Namely, there was an incorrect leading term of the Riemann surface equation which led to genus-0 functions instead of genus-1 functions responsible for the asymptotic distribution of zeros. In turn, the poles near each vertex of the asymptotic ``rectangle'' (see Fig.~\ref{fig1} right) were found incorrectly.
\end{Remark}

We now prove that the potentials $u(z)$ provided by \eqref{f1} are non-singular on the real line.

\begin{Theorem}\label{th2} The potential $u(x) = f'_1(x)+f_1^2(x)$ is non-singular at $x\in\mathbb{R}$ if
\begin{subequations}
\begin{gather}
f_1= f^{(1)}_1, \qquad n =2k, \label{nsing1}\\
f_1= f^{(2)}_1, \qquad n =2k+1, \label{nsing2}\\
f_1= f^{(3)}_1, \qquad n =2k, \label{nsing3}
\end{gather}
\end{subequations}
and $m$ is arbitrary.
\end{Theorem}

\begin{proof} Note that the functions \eqref{f1} $f_1$ can be written as
\begin{gather*}
f_1(z) = -z-v(z),
\end{gather*}
with $v=\phi_{j+1}$ as $f_1= f_1^{(j)}$, $j = 1,2,3$, $\phi_4=\phi_1$.

Then, by Corollary~\ref{cor1},
\begin{gather*}
f_1^2 = (z+v(z))^2 = \frac{c^2_{-1}}{(z-z_j)^2} + O(1), \qquad z \to z_j
\end{gather*}
for every pole $z=z_j$ of $v$. This yields
\begin{gather*}
u= f_1'+f_1^2 = \frac{-c_{-1} + c^2_{-1}}{(z-z_j)^2} + O(1), \qquad z \to z_j.
\end{gather*}
By Theorem~\ref{th1}, $c_{-1}=\pm 1$, and the pole of $u$ vanishes only for $c_{-1}=1$. Since all GHP have simple zeros, $H_{i,j}(z)= a(z-z_1) \cdots (z-z_{ij})$, we have
\begin{gather*}
c_{-1}=\mathop{\rm Res}\limits_{z=z_j} \frac{{\rm d}}{{\rm d}z}\ln \frac{H_{m+1,n}(z)}{H_{m,n+1}(z)} =
 \begin{cases} \hphantom{-}1, & H_{m+1,n}(z_j)=0,\\
-1, & H_{m,n+1}(z_j)=0,
\end{cases}
\end{gather*}
Thus one should eliminate real-valued zeros $z_j\in\mathbb{R}$ that provide $c_{-1} =-1$, i.e., the real-valued zeros of denominators in~\eqref{3p4}.

As for the cases \eqref{nsing1} and \eqref{nsing3}, the denominators in \eqref{3p4} are $H_{m,n}$ and $H_{m+1,n}$ respectively. If $n=2k$ those polynomials do not have real-valued zeros due to Theorem~\ref{th5}. The real-valued zeros of $H_{m,n+1}$ in this case provide residues $c_{-1} =1$, so that $u= f_1'+f_1^2$ is non-singular.

In the case \eqref{nsing2} we have $n=2k+1$ which yields the denominator $H_{m,n+1}$ to have no zeros on the real line. Here the zeros of the numerator provide residues $c_{-1} =1$, and $u= f_1'+f_1^2$ again is non-singular.
\end{proof}

\begin{Remark}The statement of Theorem \ref{th2} is in line with the representation \eqref{obl} of the poten\-tial~$u(x)$ by A.~Oblomkov \cite{oblom}. Due to Theorem~\ref{th5} the logarithmic derivatives in \eqref{fla-flc} are non-singular if $n$ is even in cases~a) and~c) and~$n$ is odd in case~b).
\end{Remark}

We now discuss the spectrum of Schr\"odinger operators \eqref{sch1} with non-singular potentials~\eqref{ric}
 \begin{gather*}
 L\psi = \lambda\psi, \qquad L=-\left(\frac{{\rm d}}{{\rm d}x} +f\right)\left(\frac{{\rm d}}{{\rm d}x} -f\right) = -\frac{{\rm d}^2}{{\rm d}x^2} + u(x).
 \end{gather*}
 All potentials $u$ are quadratic at infinity. Note that the quadratic potential itself $u(x) = x^2-1$ has discrete spectrum which is the set of even numbers, $\operatorname{Sp}\big({-}\frac{{\rm d}^2}{{\rm d}x^2} +x^2-1\big) = \{\lambda_k = 2k ,\, k = 0, 1, \ldots\}$. Actually, the GHP potentials demonstrate similar features.

\begin{Theorem}\label{th3} If the operator $L$ has potentials \eqref{ric} $u= f' +f^2$ with $f=f_1^{(j)}$, $j=1,2,3$ in~\eqref{f1}, then its spectrum is discrete and consists of even numbers with the exception or addition of a~finite number of terms
 \begin{subequations}
 \begin{gather}
\operatorname{Sp}(L) = \{\lambda = 2k, \, k = 0, 1, \ldots, m, m+n+1,\ldots \}, \label{speca}\\
\operatorname{Sp}(L) = \{\lambda = 2k, \, k = -m, -m+1, \ldots, 0, n+1, n+2,\ldots\}, \label{specb}\\
\operatorname{Sp}(L) = \{\lambda = 2k, \, k = -m-n, -m-n+1, \ldots,-n-1, 0, 1, \ldots\}. \label{specc}
 \end{gather}
 \end{subequations}
\end{Theorem}
\begin{proof} Since all potentials \eqref{ric} came from iterations of Darboux transformations by the dres\-sing chain \eqref{drchain}, their spectra are computed explicitly by M.~Crum's method~\cite{crum}. Namely, if the potential $u_1$ has discrete spectrum with eigenfunctions $\psi_{k}$ corresponding to distinct $\lambda_k$, $k=1,2, \ldots$, then $n$-th Darboux iteration $u_{n}$ has the form
\begin{gather*}
u_{n}(z) = u_1(z) -2\frac{{\rm d}^2}{{\rm d}z^2} \ln\mathcal{W}(\psi_{1}(z), \ldots, \psi_{n}(z)),
\end{gather*}
 and its eigenfunctions are
 \begin{gather}\label{cr1}
 \psi_{n,k}= \frac{\mathcal{W}(\psi_{1}, \ldots, \psi_{n}, \psi_{k})}{\mathcal{W}(\psi_{1}, \ldots, \psi_{n})}, \qquad k > n.
 \end{gather}
 Here $n$ should be even, else the denominator in \eqref{cr1} has real-valued zeroes and $\psi_{n,k}$ have non-integrable singularities. Since $\psi_1, \ldots, \psi_n$ are no longer eigenfunctions of new potential $u_n$, the
 spectrum of $u_n$ coincides with $\operatorname{Sp}\big({-}\frac{{\rm d}^2}{{\rm d}x^2} +u_1\big)\setminus \{\lambda_1, \ldots, \lambda_n \}$.

 In particular, put $u_1 = x^2-1$ and $\psi_1(x) = {\rm e}^{-x^2/2} H_{m+1}(x), \ldots,\psi_n(x) = {\rm e}^{-x^2/2} H_{m+n}(x)$, where $H_k$ are classical Hermite polynomials. Then
\begin{gather*}
u_n(z) = z^2 - 2\frac{{\rm d}^2}{{\rm d}z^2} \ln\mathcal{W}(H_{m+1}(z), \ldots, H_{m+n}) +2n-1 \\
\hphantom{u_n(z)}{} = z^2 - 2\frac{{\rm d}^2}{{\rm d}z^2} \ln H_{m+1,n}(z) +2n-1,
\end{gather*}
which follows from the representation \eqref{f1a}. Since $\operatorname{Sp}\big({-}\frac{{\rm d}^2}{{\rm d}x^2} +x^2-1\big) = 2\mathbb{N}$ and $\lambda_k = 2(m+k)$, we come to formula~\eqref{speca}.

The other two cases \eqref{specb} and \eqref{specc} are proved in a similar way. The spectral gap for the poten\-tial~\eqref{f1b} is $\{2m, 2(m+n)\}$ and the constant shift is $-2m$ with respect to~\eqref{f1a}. As for~\eqref{f1c}, the gap is $\{2m, 2(m+n-1)\}$ and the constant shift is $-2m-2n$. Applying Crum's formulas this gives spectra~\eqref{specb} and~\eqref{specc}. \end{proof}

Note that the case $f=f_1^{(1)}$ and the corresponding spectrum~\eqref{speca} was first found in the paper \cite{adler2} by V.E.~Adler. In the next section we reproduce Theorem~\ref{th3} by the Darboux dressing procedure of the harmonic potential.

\section{1D SUSY quantum mechanics and the PIV equation}\label{sect4}

The idea to factorize quantum Hamiltonians and get supersymmetric potentials dates back to the pioneering paper by E.~Witten \cite{W}. Later came examples of one-dimensional realizations of this idea with the simplest polynomial Heisenberg algebras. A~connection between the harmonic oscillator and supersymmetric (SUSY) partner potentials generated by this algebra has been long known. Recently these potentials were identified with rational solutions of the PIV equation. Here we follow the papers by D.~Bermudes and D.J.~Fern\'andez~C.~\cite{BF1,BF2} describing this application.

Starting with two Schr\"odinger operators \eqref{sch}, say $L_j$ and $L_{j+1}$, one can factorize them as shown in Section~\ref{sect1}
\begin{gather*}
L_j = A_j^+A_j^- + \epsilon_j, \qquad L_{j+1} = A_{j}^-A_{j}^+ + \epsilon_j,
\end{gather*}
where
\begin{gather}\label{aop}
A_j^+ = \frac{{\rm d}}{{\rm d}x} + f_j(x), \qquad A_j^- = \frac{{\rm d}}{{\rm d}x} - f_j(x)
\end{gather}
and $u_j \mapsto u_j + \epsilon_j$, $\alpha_j = 2(\epsilon_j - \epsilon_{j+1})$ in \eqref{ric}, \eqref{drchain}.

Then an intertwining relation holds, namely, $L_{j+1}A_j^+ = A_j^+L_j$ which generates the $k$-th order intertwining operators
\begin{gather}
L_k B_k^+ = B_k^+L_1, \qquad L_1 B_k^- = B_k^-L_k, 
\qquad B_k^+=A_k^+ \cdots A_1^+, \qquad B_k^-=A_1^- \cdots A_k^-. \label{inter2}
\end{gather}
This represents the standard SUSY algebra
\begin{gather*}
 [\mathcal{Q}_a, \mathcal{H}] = 0, \qquad \{\mathcal{Q}_a, \mathcal{Q}_b\} = \delta_{ab}\mathcal{H}, \qquad a,b =1,2, \\
\mathcal{Q}_1= \frac{1}{\sqrt{2}}\begin{pmatrix}0 & B_k^+ \\ B_k^-& 0\end{pmatrix}, \qquad
\mathcal{Q}_2= \frac{1}{i\sqrt{2}}\begin{pmatrix}0 & B_k^+ \\ -B_k^-& 0\end{pmatrix}, \qquad
\mathcal{H}= \frac{1}{\sqrt{2}}\begin{pmatrix} B_k^+B_k^-& 0 \\ 0 & B_k^-B_k^+ \end{pmatrix}.
\end{gather*}
where $\{ \, , \}$ is the anticommutator and $\mathcal{H}$ is the Hamiltonian with superpotential partner $u_k=u_1 + 2(f_1+ \dots + f_k)'$ of the initial potential~$u_1$~\cite{BF1}.

Similarly, the polynomial Heisenberg algebra for the Hamiltonian \eqref{sch} is formed by the dressing chain operators
\begin{gather*}
\big[L,B_k^\pm\big]=\pm B_k^\pm, \qquad \big[B_k^-,B_k^+\big]=N_k(L+I) - N_k(L), \qquad N_k(L)=\prod_{j=1}^k(L-\epsilon_j).
\end{gather*}
For $k=3$ this algebra can produce new solutions of the PIV equation starting from the known ones (see \cite{BF2}). Taking a closure condition $L_4=L_1-I$ we come to the dressing chain \eqref{adl} with some~$\alpha_1$,~$\alpha_2$ and $\alpha_3$. An example is given at the end of this section.

First we start with the trivial solution of the ``$-2x$ hierarchy'' \eqref{hier}, which corresponds to the harmonic potential $u_0(x) = x^2 {-1}$. The eigenfunctions of the Schr\"odinger operator with the harmonic potential
\begin{gather*}
L\psi_{0}= {2}\epsilon\psi_{0}
\end{gather*}
are written explicitly
\begin{gather} \label{wh}
\psi_{0}(x) =\sqrt{\pi} \big((1- c + {\rm i}(1+c){\rm e}^{{\rm i}\pi\epsilon}\big)D_{\epsilon -1}\big(\sqrt{2} x\big)
 - 2(1+c){\rm e}^{{\rm i}\pi\epsilon /2}\sin(\pi\epsilon) \Gamma(\epsilon) D_{-\epsilon }\big({\rm i}\sqrt{2} x\big),\!\!\!
 \end{gather}
where $D_\nu(x)$ is the Weber--Hermite function, $\Gamma$ is the gamma function and $c$ is an arbitrary complex-valued constant.

Apply now the ladder operators \eqref{inter2} to the basic eigenfunction \eqref{wh}. The first SUSY partner potential for $k=1$ to $u_0(x) = x^2 {-1}$ has the form
\begin{gather*}
u_1(x) = x^2 - {2}\frac{{\rm d}^2}{{\rm d}x^2}\ln \psi_0(x) {-1}.
\end{gather*}
Similarly, the $k$-th potential becomes \cite{BF1}
\begin{gather}\label{spartner}
u_k(x) = x^2 - {2}\frac{{\rm d}^2}{{\rm d}x^2}\ln \mathcal{W} (\psi_{01}, \psi_{02}, \ldots , \psi_{0 {k}} ) {-1},
\end{gather}
where $\mathcal{W}$ is the Wronskian of eigenfunctions \eqref{wh} $\psi_{0j}(x)$ with parameters $\epsilon = \epsilon_j$.

This gives a way to construct a set of rational potentials associated with GHP. Take para\-me\-ters~$\epsilon_j$ to be integers, because in this case the Weber--Hermite functions in \eqref{wh} become Hermite polynomials multiplied by the exponential ${\rm e}^{-x^2/2}$. Namely, if we put
\begin{gather*}
\epsilon_1 = m+ 1, \qquad \epsilon_j=m+j, \qquad c_j =-1,\qquad j=2,3, \ldots, n, \qquad m,n \in \mathbb{N},
\end{gather*}
the second term in formula \eqref{wh} vanishes and the potential \eqref{spartner} takes the form
\begin{gather*}
 u_n(x) = x^2 - 2\frac{{\rm d}^2}{{\rm d}x^2}\ln \mathcal{W}\left(H_m(x), H_{m+1}(x), \ldots , H_{m+n}(x) \right) +2n -1\\
\hphantom{u_n(x)}{} = x^2 - 2\frac{{\rm d}^2}{{\rm d}x^2}\ln H_{m,n+1}(x) + 2n-1.
\end{gather*}
Unfortunately, this choice of $\psi_0$ leads to singularities in the intermediate eigenfuntions \eqref{cr1} of the potential \eqref{spartner}
\begin{gather}\label{cr2}
\psi_{k,j} = \frac{ \mathcal{W} (\psi_{01}, \psi_{02}, \ldots , \psi_{0 j-1} )}{ \mathcal{W}\left(\psi_{01}, \psi_{02}, \ldots , \psi_{0 j}\right)}, \qquad j=2,3, \ldots, k.
\end{gather}
To avoid this and make the formal dressing procedure correct, we use the following theorem proved by D.~Bermudez and D.J.~Fern\'andez~C.\ in~\cite{BF2} providing non-singularity of eigenfunctions.

\begin{Theorem}\label{th4} Schr\"odinger operators \eqref{sch} with potentials \eqref{spartner} are mo\-no\-dromy-free. All potentials \eqref{spartner} and eigenfunctions \eqref{cr2} are non-singular for $x\in \mathbb{R}$ if the dressing parameters satisfy the conditions
\begin{gather}
\epsilon_k < \epsilon_{k-1} < \dots < \epsilon_1 < 0, \label{epsj}\\
|c_{2j}| <1, \qquad |c_{2j+1}| > 1, \qquad j =0, 1, \ldots , k, \label{cj}
\end{gather}
where $\epsilon_j = \epsilon$ and $c_j = c$ in \eqref{wh}.
\end{Theorem}

\begin{proof} Since the Weber--Hermite functions $D_\nu(z) $ in \eqref{wh} are entire functions, so is the Wrons\-kian~\eqref{spartner}. Thus its logarithmic derivative is a~meromorphic function. Moreover, any solution~$\psi(z,\lambda)$ of the Schr\"odinger equation with potential~\eqref{spartner} can be found via a~{\em finite} dressing chain of~$k$ operators~$A_j^\pm$~\eqref{aop}, i.e., a~finite number of Darboux transformations. This yields trivial monodromy of~$\psi(z,\lambda)$.

It is easy to check that the function \eqref{wh} has no real-valued zeros if $\epsilon < 0$ and $|c| < 1$. The remaining inequalities \eqref{epsj} and \eqref{cj} are proved by induction (see~\cite{BF1} for details). \end{proof}

Apply now Theorem \ref{th4} to get non-singular functions \eqref{cr2}. Take parameters $\epsilon_j$ and $c_j$ in the form
\begin{gather*}
\epsilon_1 = -n, \!\!\qquad \epsilon_j=-n-j,\!\! \qquad c_{2j} =0,\!\!\qquad c_{2j-1}=\infty,\!\! \qquad j=1,2, \ldots, m-1,\!\! \qquad m,n \in \mathbb{N}.
\end{gather*}
In the case of odd $j$ the eigenfunctions \eqref{wh} are normalized as $\psi_{0j}=D_{-\epsilon}\big({\rm i}\sqrt{2} x\big)$. This forces the first term in~\eqref{wh} to be zero while the second term turns into the Hermite polynomial ${\rm e}^{x^2/2}H_{n+j}({\rm i}x)$. Thus the potential~\eqref{spartner} becomes
\begin{gather}
 u_n(x) = x^2 - 2\frac{{\rm d}^2}{{\rm d}x^2}\ln \mathcal{W} (H_n({\rm i}x), H_{n+1}({\rm i}x), \ldots , H_{n+m-1}({\rm i}x)) -2m-1 \nonumber\\
 \hphantom{u_n(x)}{} = x^2 - 2\frac{{\rm d}^2}{{\rm d}x^2}\ln H_{m,n}(x) - 2m-1, \label{cr3}
\end{gather}
because $H_{k,j}({\rm i}x)={\rm i}^{kj}H_{j,k}(x)$.

As shown in \cite{BF2}, the spectrum of the potential \eqref{cr3} is
\begin{gather*}
\operatorname{Sp}\left(-\frac{{\rm d}^2}{{\rm d}x^2}+u_n\right)= \{2\epsilon_{m-1}-2, \ldots, 2\epsilon_1-2\}\cup \{0, 2, 4, \ldots\},
\end{gather*}
which coincides with formula \eqref{f1c} found in Theorem \ref{th3}.

Finally we illustrate the dressing of the ``$-2x$ hierarchy'' by the ladder operators \eqref{aop},
 \eqref{inter2} with singular eigenfunctions for the case $k=3$
\begin{gather}\label{ladder3}
B_3^\pm = \left(\frac{{\rm d}}{{\rm d}x} \pm f_1\right)\left(\frac{{\rm d}}{{\rm d}x} \pm f_2\right)\left(\frac{{\rm d}}{{\rm d}x} \pm f_3\right).
\end{gather}
 The closure condition $L_4 = L_1 - I$ leads to the system
\begin{gather*}
f_1' + f_2' = f_1^2 - f_2^2 - 2(\epsilon_1 - \epsilon_2), \\
f_2' + f_3' = f_2^2 - f_3^2 - 2(\epsilon_2 - \epsilon_3), \\ 
f_3' + f_1' = f_3^2 - f_1^2 - 2(\epsilon_3 - \epsilon_1 {-1}),
\end{gather*}
 which is equivalent to the sPIV system \eqref{adl}, where
 \begin{gather*}
 \alpha_1 = 2(\epsilon_1 - \epsilon_2), \qquad \alpha_2 = 2(\epsilon_2 - \epsilon_3), \qquad \alpha_3 = 2(\epsilon_3 - \epsilon_1 {-1}).
 \end{gather*}

\begin{figure}[t]\centering
\includegraphics[width=70mm]{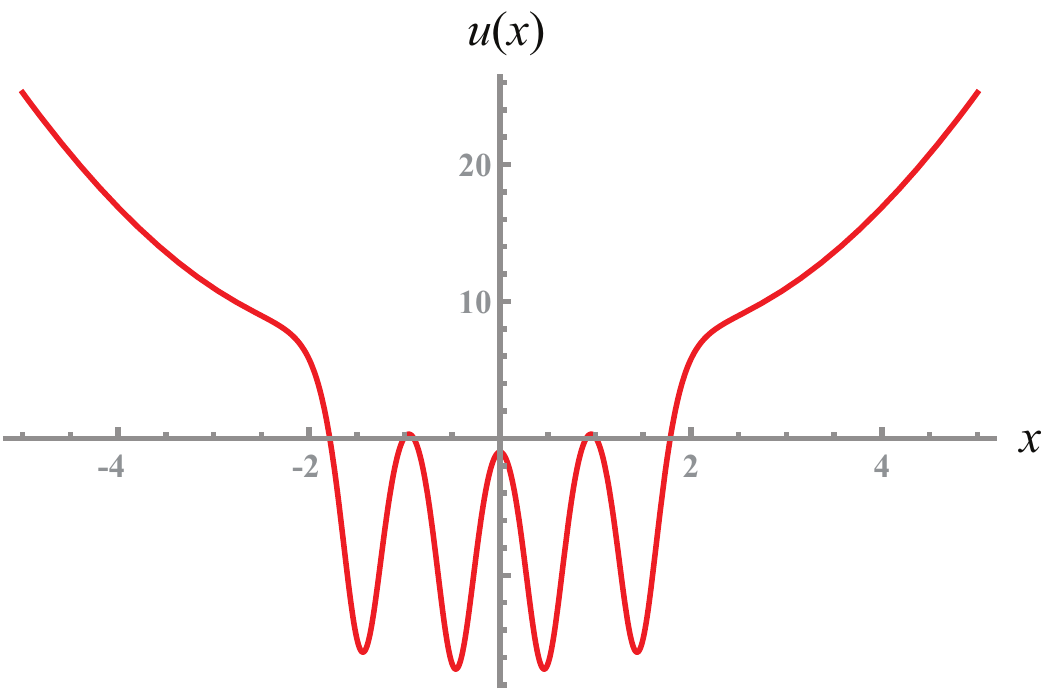} \qquad\includegraphics[width=70mm]{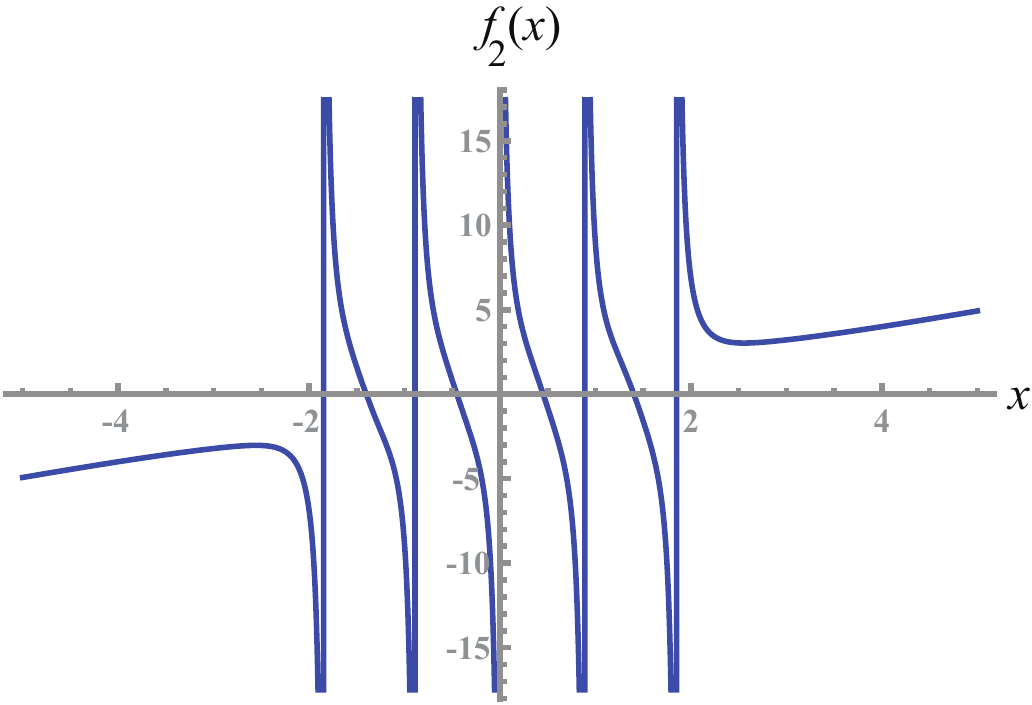}

\caption{Non-singular potential $u(x)=f_2'(x)+f^2_2(x)$ generated by ladder operator~\eqref{ladder3} (left) and its singular component~\eqref{lf2} $f_2(x)$.} \label{fig2}
\end{figure}

 Choosing {integer} values of $\epsilon_1$, $\epsilon_2$ and $\epsilon_3$ it is easy to reproduce the generalized Hermite polynomials which form rational solutions of PIV~\eqref{3p4}. For example, the choice
 \begin{gather*}
 {\epsilon_3 = -3, \qquad \epsilon_1 = 1, \qquad \epsilon_2 = 4,}
 \end{gather*}
yields $m=4$ and $n=3$. This provides a~rational solution of PIV in the form \eqref{3p4}
 \begin{gather*}
 \phi_2(x) = \frac{{\rm d}}{{\rm d}x}\ln \frac{H_{4,3}(x)}{H_{4,4}(x)} =\frac{64 x^3 \big(7875 - 900 x^4 + 720 x^8 + 64 x^12\big)}{23625 + 16 x^4 \big(7875 - 450 x^4 + 16 x^8 \big(15 + x^4\big)\big)} \\
\hphantom{\phi_2(x) =}{} +\frac{24 x \big({-}225 + 2 x^2 \big({-}75 - 60 x^2 + 120 x^4 - 40 x^6 + 16 x^8\big)\big)}{ 675 + 4 x^2 \big({-}675 - 225 x^2 + 4 x^4 \big({-}30 + 45 x^2 - 12 x^4 + 4 x^6\big)\big)}.
\end{gather*}
Since $n$ is odd, by Theorem \ref{th2} one should choose a non-singular potential from the func\-tion~\eqref{f1}, \eqref{nsing2}
 \begin{gather}
 f_2(x) = f_1^{(2)}=x + \frac{{\rm d}}{{\rm d}x}\ln \frac{H_{5,3}(x)}{H_{4,4}(x)} \nonumber\\
 =x+\frac{1}{x}-\frac{64 x^3 \big(7875-900 x^4+720 x^8+64 x^{12}\big)}{23625+16 x^4 \big(7875+2 x^4 \big({-}225+8 x^4 \big(15+x^4\big)\big)\big)}\label{lf2}\\
\hphantom{=}{} +\frac{4 x \big({-}7875+4 x^2 \big({-}4725+2025 x^2-4200 x^4+2700 x^6-720 x^8+112 x^{10}\big)\big)}{23625+2 x^2 \big({-}7875+2 x^2 \big({-}4725+2 x^2 \big(675+2 x^2 \big({-}525+270 x^2-60 x^4+8 x^6\big)\big)\big)\big)}.\nonumber
\end{gather}
The non-singular potential has the form
\begin{gather*}
u(x) = f_2'(x) + f_2^2(x)=x^2-1 + R(x),
\end{gather*}
where $R$ is rational function, $R(x) = O\big(x^{-2}\big)$, $x\to\infty$. The functions $u$ and $f_2$ are plotted in Fig.~\ref{fig2}. The spectrum of the Schr\"odinger operator $L$ with this potential is
\begin{gather*}
\operatorname{Sp}(L)= \{-8,-6,-4,-2\}\cup\{8,10,12, \dots\}.
\end{gather*}

\subsection*{Acknowledgements} The work has been supported by Russian Scientific Foundation grant 17-11-01004. The author also is grateful to the referee remarks which helped to improve the paper.

\pdfbookmark[1]{References}{ref}
\LastPageEnding

\end{document}